\documentclass[10pt,journal]{IEEEtran} 
\usepackage{graphicx}
\usepackage{amsmath}
\usepackage{amssymb}
\usepackage[noend]{algpseudocode}
\usepackage{algorithm}
\usepackage{cite}

\usepackage{subcaption}
\usepackage{color}
\usepackage{textcomp}
\usepackage{cite}
\usepackage{amsthm}

\newtheorem{prop}{Proposition}

\begin{document}
	\title{Beam Acquisition and Training in Millimeter Wave Networks 
		with Narrowband Pilots
	}
\author{Hao~Zhou,
	Dongning~Guo,~\IEEEmembership{Senior Member,~IEEE,}
	and~Michael~L.~Honig,~\IEEEmembership{Fellow,~IEEE}
	\thanks{The authors are with the Department
		of Electrical and Computer Engineering, Northwestern University, Evanston, IL 60208 USA (email: haozhou2015@u.northwestern.edu; dGuo@northwestern.edu; mh@eecs.northwestern.edu).}
	\thanks{The work was presented in part at Asilomar 2018 \cite{hao2018initialaccess}. This work was supported in part by a gift from Futurewei Technologies.}}

	\maketitle

	\begin{abstract}
		This paper studies initial beam acquisition in a millimeter wave network consisting of multiple access points (APs) and mobile devices. A training protocol for joint estimation of transmit and receive beams is presented with a general frame structure consisting of an initial access sub-frame followed by data transmission sub-frames. During the initial subframe, APs and mobiles sweep through a set of beams and determine the best transmit and receive beams via a handshake. All pilot signals are narrowband (tones), and the mobiles are distinguished by their assigned pilot frequencies. Both non-coherent and coherent beam estimation methods based on, respectively, power detection and maximum likelihood (ML) are presented. To avoid exchanging information about beamforming vectors between APs and mobiles, a local maximum likelihood (LML) algorithm is also presented. An efficient fast Fourier transform implementation is proposed for ML and LML to achieve high-resolution. A system-level optimization is performed in which the frame length, training time, and training bandwidth are selected to maximize a rate objective taking into account blockage and mobility. Simulation results based on a realistic network topology are presented to compare the performance of different estimation methods and {training codebooks}, and demonstrate the effectiveness of the proposed protocol.
	\end{abstract} 
	

	\IEEEpeerreviewmaketitle
	
	\section{Introduction}
	Fifth generation (5G) wireless communication networks are expected to provide ubiquitous connectivity and increased throughput to support the 
	increasing demand for mobile data services \cite{wong2017key}. As centimeter wave (especially sub-6 GHz) bands become crowded, millimeter wave (mmWave) 
	bands are viewed as an important resource for satisfying projected service objectives. Recent channel measurement campaigns at mmWave frequencies have indicated that while the attenuation is relatively high, the channel typically consists of a small number of propagation paths \cite{rappaport2017overview}. Beamforming and combining with a large number of antennas, known as massive multi-input multi-output (MIMO), can therefore focus the signals along the strong paths to maintain a desired signal-to-noise ratios (SNR) at the receiver. Designing the transmit and receive beams requires channel state information (CSI), which 
	is obtained through training. 
	
	In mmWave systems, channel estimation often takes the form of \emph{beam training}, which, by sending training signals, estimates the key parameters of the channel including the number of paths, spatial directions of the paths, and the path gains. Depending on whether the CSI is available \emph{a priori}, beam training can be classified into initial access \cite{desai2014initial,barati2016initial,giordani2016initial} and beam tracking \cite{zhang2016tracking,zhang2016mobile,zhu2017abp,Palacios2017Tracking}. The initial access aims to establish a communication link without prior knowledge of the channel. Because of the high attenuation at mmWave frequency bands, broadcasting omni-directional training signals for discovery of access points (AP) and channel sensing is often inadequate. Due to mobility and blockage, old paths may fade and new paths may emerge, which requires repeated training. By contrast, beam tracking assumes the existence of a communication link, and the goal is to track the deviation of the paths and refine the transmit/receive beams. 
	
	Dense deployment of mmWave APs is expected 
	to overcome blockage and improve 
	coverage \cite{Simic2017coverage}. With many APs and mobiles in an area, interference coordination becomes important for both beam training and data transmission. Both protocols and algorithms have been considered for training and data transmission in \cite{gonzalez2018channel,sun2019beam,alkhateeb2015limited}, with a single AP.  
	
	In this paper, we focus on the initial access problem. 
	Specifically, we consider the design of a training protocol for joint beam acquisition and tracking.  
	We try to address the question of how much training overhead is needed for an mmWave system with multiple APs and mobiles. Our main contributions are summarized as follows:
	
	1) We propose a system-level protocol and a frame structure for establishing connections between multiple APs and mobiles. 
	Because the directions of propagation paths are essentially identical across a typical mmWave band, we use narrowband signals (tones) for training. 
	The estimated beamforming and combining filters are then used for wideband data transmission. This narrowband scheme effectively avoids mutual interference and provides a high SNR. Different {training 
		codebooks} (exhaustive sweeping, compressive sensing, etc.) as well as channel estimation methods can be incorporated into the protocol.
	
	2) We present three channel estimation methods, namely the max power (MP) method, the maximum likelihood (ML) method, and the local maximum likelihood (LML) method. A low-complexity fast Fourier transform (FFT) implementation of the ML and LML methods is proposed to obtain a near-optimal estimate, regardless of the 
	{training codebook}. In particular, with exhaustive sweeping, we show that to minimize the training error no pilot repetition (per slot) is needed for ML. 
	We compare the performance of these methods. 
	
	3) We perform a system-level analysis and determine the optimal training overhead. The overhead is determined by optimizing system parameters including the frame length, training duration, and training bandwidth. The objective is to 
	maximize the long-term network throughput, accounting for 
	random blockage and link-level training error. The solution indicates that the training overhead is around $5 \%$ under typical scenarios; however, with severe frequent blockage and worse channel conditions, the overhead increases to above $10\%$, in which case selecting a training scheme with lower overhead becomes important. 
	
	The paper is organized as follows. In the next section we discuss related work. 
	In \mbox{Section \ref{sec:sysModel}}, we introduce the channel model, system model, and hybrid beamforming structures. In Section \ref{sec:access-protocol}, we present the narrowband training protocol and the frame structure. In \mbox{Section \ref{sec:est-method}}, we discuss the three beam training methods. In Section \ref{sec:sys}, we analyze the system-level performance and optimize the key parameters of the training protocol. Section \ref{sec:simu} presents simulation results and Section \ref{sec:conclusion} presents our conclusions.
	
	
	
	\section{Related Work} \label{sec:related-work}
	{Training protocols for mmWave initial access have been considered for a single AP with a single mobile in \cite{alkhateeb2015limited,hur2013millimeter,zhao2017multiuser} and multiple mobiles in \cite{Sun2018Hybridbeamforming,zhao2017tone,marzi2016compressive}. To avoid inter-user interference and pilot contamination, a tone-based training scheme is proposed in \cite{zhao2017multiuser}, where each mobile is assigned a distinct frequency tone during training. An algorithm is proposed to estimate the AoAs at APs and mobiles for a single cell with multiple users and design the analog and digital zero-forcing precoders. In \cite{zhao2017tone}, the tone-based AoA estimation method is considered for a multi-cell network, where the mobiles have a single antenna and the APs use digital beamforming. Those papers analyze the achievable rate for proposed channel estimation and beamforming algorithms. Here we consider the use of tones for user identification and acquisition in a network with a large number of APs and mobiles. Instead of fixing user association, as in \cite{zhao2017multiuser,zhao2017tone}, the mobiles acquire their APs via a sweeping protocol. We analyze the training overhead for the proposed protocol and optimize the associated parameters. 
		
	Related work on estimating the key parameters of mmWave channels, including the number of paths, spatial directions, and path gains, is presented in \cite{alkhateeb2015limited,zhao2017multiuser,zhao2017tone,zhang2016mobile,marzi2016compressive}. The MP method for estimating spatial directions has been considered in \cite{alkhateeb2015limited,zhao2017multiuser,zhao2017tone}. It is simple and robust with respect to system impairments, but the training overhead scales with increasing estimation resolution. In \cite{zhang2016mobile}, an ML-based channel estimation method is proposed, and the resulting non-linear least squares problem is solved with the Levenberg-Marquardt algorithm. The ML method is also considered in \cite{marzi2016compressive}, where random beamforming is used for training. The AoAs and AoDs are estimated at APs with uplink feedback using the Newton's method. Compared with MP estimation, ML methods are more robust with respect to noise and can obtain high-resolution estimates with less training, but require knowledge of both beamforming and combining vectors to compute the estimate. 
	The LML method we present only needs to know the receiver filters. It does not require the exchange of filters between APs and mobiles. We also propose a method for computing the ML and LML estimates using FFTs.
		
	System-level analysis of mmWave networks has been studied in \cite{Ghadikolaei2015mmWave,marzi2016compressive,li2017design,li2017onthebeamformed,bai2015coverage,zhao2018multi}. In \cite{marzi2016compressive}, an SNR threshold required for successful estimation is derived using the Ziv-Zakai bound (ZZB) and the Cram\'{e}r-Rao bound (CRB). In \cite{Ghadikolaei2015mmWave}, a two-step initial access procedure is proposed, which exploits the features of omni-directional microwave and directional mmWave signals. The combination of directional beamforming and omni-directional transmission for initial access is also considered in \cite{li2017design}. In \cite{li2017onthebeamformed}, different codebooks are evaluated in terms of access latency and overhead. Design insights are also provided on the beam width and the number of simultaneous beams. In \cite{bai2015coverage}, coverage and rate performance are analyzed assuming an ideal sectored beam pattern in the absence of training error. In contrast, we maximize the long-term system throughput taking into account both link-level training error and random blockages.} 
	
	\section{System Model}\label{sec:sysModel}
	\begin{figure}[!t]
		\centering
		\includegraphics[width=0.8\linewidth]{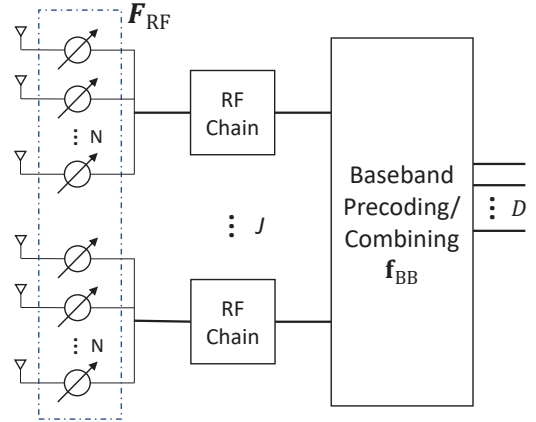}
		\caption{Example of a sub-connected hybrid transmitter/receiver structure at an AP; $D\leq J$ is the total number of data streams. The same structure applies to mobiles.}
		\label{fig:rfchain}
	\end{figure}
	Consider a mmWave system with $L$ APs and $K$ mobile devices. As illustrated in Fig. \ref{fig:rfchain}, every transceiver is equipped with $J$ radio-frequency (RF) chains. All transceivers are assumed to adopt the partially connected hybrid beamforming architecture. Each RF chain is connected to a sub-array of phased antennas through constant-modulus phase shifters \cite{heath2016overview}. We assume each antenna sub-array at an AP consists of $N$ antennas, and each sub-array at a mobile consists of $M$ antennas. 
	
	The hybrid precoding structure was introduced in \cite{zhang2005variable}. Due to the low hardware complexity, it has been extensively studied for massive MIMO and mmWave systems \cite{alkhateeb2014channel,adhikary2013jointspatial,sohrabi2017hybrid,gao2016energy,ayach2012thecapacity}. In this paper, we adopt the beam steering method \cite{ayach2012thecapacity}, where signals are transmitted by steering beams to the direction of the strongest path. Beam steering is simple and widely used in practice \cite{sadhu2017phasearray}. It has been shown that when the total number of antennas $NJ,MJ\to\infty$, beam steering is asymptotically optimal for a single data stream \cite[Corollary 4]{ayach2012thecapacity}. 
	
	\subsection{Channel Model} \label{sec:channelModel}
	We consider a MIMO multipath channel model where the channel has a small number of propagation paths \cite{rappaport2017overview}. Due to the small form factor of mmWave antennas, for a particular AP and mobile pair, we assume that the channels across different sub-array combinations share the same directions and path loss, but with different delays \cite{heath2016overview}. The downlink virtual channel from a particular AP to a mobile is
	{\begin{align}\label{eq:channel}
		\mathbf{H} = \sum_{s=1}^{S}\alpha_{s} \mathbf{u}(\boldsymbol{\theta}_s)\mathbf{a}^H(\boldsymbol{\phi}_s),
		\end{align}
		where $S$ denotes the total number of propagation paths, $\boldsymbol{\theta}_s$ and $\boldsymbol{\phi}_s$ denote respectively the angle of arrival (AoA) and the angle of departure (AoD), $\mathbf{u}\in\mathbb{C}^{JM}$ and $\mathbf{a}\in\mathbb{C}^{JN}$ denote the antenna response functions at the mobile and the AP, respectively, and we define $\alpha_{s} =\nu\tilde{\alpha}_s$ with $\nu=\sqrt{NMJ^2/S}$ denoting the antenna gain and $\tilde{\alpha}_s\in\mathbb{C}$ denoting the path gain. The path gain $\tilde{\alpha}_s$ includes both path loss and delay with $\mathbb{E}[|\tilde{\alpha}_s|]=\bar{\alpha}_s$. }
	
	The response functions $\mathbf{u}$ and $\mathbf{a}$ depend on the geometry 
	of the antenna arrays and are well defined for any particular antenna configuration. For concreteness, we focus on two commonly used array structures in practice: uniform linear array (ULA) and uniform planar array (UPA), which are illustrated in Fig. \ref{fig:antenna}. 
	
	\begin{figure}[!t]
		\centering
		\includegraphics[width=0.95\linewidth]{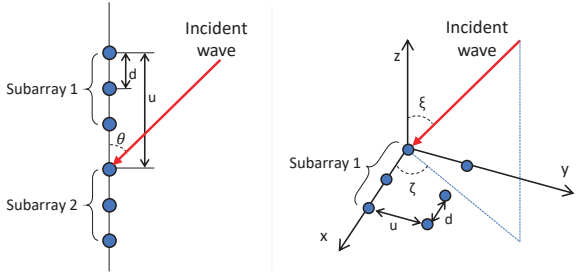}
		\caption{Configurations of two sub-arrays, each with three elements. Left: ULA. Right: UPA.}
		\label{fig:antenna}
	\end{figure}
	
	We take $\mathbf{u}(\boldsymbol{\theta})$ as an example. A similar structure applies to $\mathbf{a}(\boldsymbol{\phi})$. To simplify notation, we first define a discrete Fourier transform (DFT) vector of length $M$:
	\begin{align}\label{eq:ang_vec}
	\mathbf{e}(\vartheta;M) =\sqrt{1/M}\left[1,e^{j\vartheta},\cdots,e^{j(M-1)\vartheta}\right]^T.
	\end{align}
	
	The AoA of a ULA is fully characterized by a single angle $\theta$ representing the incident wave and the line of the antennas. The angular response of an antenna sub-array is then {$\mathbf{e}\left(\frac{2\pi d \sin\theta}{c/f_c};M\right)$} 
	and the angular response of the ULA is then
	\begin{align} \label{eq:ula}
	\begin{split}
	\mathbf{u}({\theta}) = \mathbf{e}\left(\frac{2\pi u \sin\theta}{c/f_c};J\right) \otimes \mathbf{e}\left(\frac{2\pi d \sin\theta}{c/f_c};M\right),
	\end{split}
	\end{align}
	where $\otimes$ denotes the Kronecker product,\footnote[1]{$[a_1,\dots, a_J]^T\otimes[b_1,\dots, b_M]^T = [a_1b_1,\dots, a_1b_M, \dots, a_Jb_M]^T$.} $c$ denotes the light of speed, $f_c$ denotes the carrier frequency, $d$ denotes the antenna element spacing within each sub-array, and $u$ denotes the distance between the first elements of adjacent sub-arrays. 
	
	In contrast, the AoA of a UPA is characterized by two angles $\boldsymbol{\theta}=[\zeta,\xi]$, where $\zeta$ denotes the azimuth angle and $\xi$ denotes the elevation angle. The antenna response function is
	\begin{align}\label{eq:upa}
	\begin{split}
\mathbf{u}(\boldsymbol{\theta}) = \mathbf{e}\left(\frac{\sin(\xi)\sin(\zeta)}{c/(2\pi uf_c)};J\right) \otimes \mathbf{e}\left(\frac{\sin(\xi)\cos(\zeta)}{c/(2\pi df_c)};M\right),
	\end{split}
	\end{align}
	where $d$ and $u$ denote the antenna element spacing on the x-axis and y-axis, respectively. In either case (ULA or UPA), we therefore have 	
	\begin{align}\label{eq:upa_ula}
	\begin{split}
	\mathbf{u}(\boldsymbol{\theta}) = \mathbf{e}\left(\vartheta_1;J\right) \otimes \mathbf{e}\left(\vartheta_2;M\right)  ,
	\end{split}
	\end{align}
	where $\vartheta_1$ and $\vartheta_2$ are defined as in \eqref{eq:ula} or \eqref{eq:upa} depending on whether the arrays take the form of a ULA or UPA. 
	
	\subsection{Signal Model}
	Assuming no inter-symbol interference, the time index of all signals can be suppressed. Based on \eqref{eq:channel}, the downlink channel from AP $l$ to mobile $k$ is
	\begin{align} \label{eq:channel_withIndex}
	\mathbf{H}_{l,k} = \sum_{s=1}^{S_{l,k}}\alpha_{s,l,k} \mathbf{u}(\boldsymbol{\theta}_{s,l,k})\mathbf{a}^H(\boldsymbol{\phi}_{s,l,k}).
	\end{align}
	We let AP $l$ transmit a single stream of baseband symbols. The downlink baseband received signal at mobile $k$ is
	\begin{align} \label{eq:rcvSig_dLink}
	y_k = \mathbf{w}^H_k\left(\sum_{l=1}^{L}  \mathbf{H}_{l,k}\mathbf{f}_lx_l\right) + \mathbf{w}^H_k\mathbf{n}_k,
	\end{align}
	where $\mathbf{f}_l\in\mathbb{C}^{JN}$ and $\mathbf{w}_k\in\mathbb{C}^{JM}$ denote the hybrid beamforming and combining vectors,  $x_l\in\mathbb{C}$ denotes the downlink symbol sent by AP $l$, and $\mathbf{n}_k\sim \mathcal{CN}(0,\sigma_n^2\mathbf{I}_{JM})$ denotes the additive white Gaussian noise. 
	
	We assume time division duplex (TDD) mode. In the uplink, the baseband received signal at AP $l$ is
	\begin{align} \label{eq:rcvSig_upLink}
	r_l = \mathbf{g}^H_l\left(\sum_{k=1}^{K}  \mathbf{H}^H_{l,k}\mathbf{v}_ks_k\right) + \mathbf{g}^H_l\tilde{\mathbf{n}}_l,
	\end{align}
	where $\mathbf{g}_l\in\mathbb{C}^{JN}$, $\mathbf{v}_k\in\mathbb{C}^{JM}$, $s_k\in\mathbb{C}$, and \mbox{$\tilde{\mathbf{n}}_l\sim \mathcal{CN}(0,\sigma_n^2\mathbf{I}_{JN})$} denote the hybrid beamforming vector, hybrid combining vector, uplink symbol, and additive noise, respectively. 
	
	All the hybrid beamforming/combining vectors $\mathbf{f}_l,\mathbf{g}_l,\mathbf{v}_k$, and $\mathbf{w}_k$ are the composition of a digital baseband filter with an analog precoder. For example, $\mathbf{f} = \mathbf{F}_\text{RF}\mathbf{f}_\text{BB}$, where \mbox{$\mathbf{f}_\text{BB}\in\mathbb{C}^{J}$} denotes the digital baseband precoder and the analog precoder $\mathbf{F}_\text{RF} = \text{diag}({\mathbf{a}}_1,{\mathbf{a}}_2,\dots,{\mathbf{a}}_J)$. The $i$-th diagonal block \mbox{$\mathbf{a}_i=\mathbf{a}(\boldsymbol{\phi}_i)$} corresponds to the phase shifters in the $i$-th sub-array. 
	
	
	We adopt beam steering, as in \cite{ayach2012thecapacity}, where 
	all the sub-arrays of a mobile point towards a common direction. Specifically, the steering vectors $\mathbf{w},\mathbf{v}$ take the same form as \eqref{eq:upa_ula},
	\begin{align}\label{eq:shortHBF}
	\mathbf{w} = \sqrt{\tilde{\rho}}\mathbf{u}(\boldsymbol{\theta}) =  \sqrt{\tilde{\rho}}\mathbf{e}\left(\vartheta_1;J\right) \otimes \mathbf{e}\left(\vartheta_2;M\right),
	\end{align}
	where $\tilde{\rho}$ is a power control variable. We can also write \eqref{eq:shortHBF} as $\mathbf{w} = \mathbf{W}_\text{RF}\mathbf{w}_\text{BB}$ with the digital precoder $\mathbf{w}_\text{BB} = \sqrt{\tilde{\rho}}\mathbf{e}\left(\vartheta_1;J\right)$ and where the diagonal blocks of $\mathbf{W}_\text{RF}$ are $\bar{\mathbf{w}}_i=\mathbf{e}\left(\vartheta_2;M\right)$, \mbox{$1\leq i\leq J$.} Then each mobile has only three design parameters: angles $\vartheta_1$, $\vartheta_2$, and \mbox{power $\tilde{\rho}$.}
	
	If an AP serves only a single mobile, then it steers the beam towards that mobile, so that \eqref{eq:shortHBF} applies to its beamforming vectors $\mathbf{f}$ and $\mathbf{g}$. If an AP simultaneously serves multiple mobiles, we steer different sub-arrays towards different mobiles. Then the diagonal blocks of the analog precoder $\mathbf{a}_i=\mathbf{a}\left(\boldsymbol{\phi}_i\right)$ no longer take a common parameter $\boldsymbol{\phi}$. Also, designing the digital precoder with equal power allocation, as in \eqref{eq:shortHBF}, may not be optimal in general. Therefore, there are $2J$ design parameters: steering angles $\boldsymbol{\phi}_1,\cdots, \boldsymbol{\phi}_J$, and the digital precoder $\mathbf{u}_\text{BB}$.
	
	\section{Multiple Access Protocol}\label{sec:access-protocol}
	In this section, we propose a multiple access protocol for a mmWave network consisting of multiple APs and mobiles. The goal is to establish communication links, design beamforming and combining filters, and maintain connections with occurrences of blockage. We first present the frame structure and the protocol. We then focus on the initial access period and present a narrowband training procedure. 
	
	\subsection{Frame Structure and Multiple Access Protocol}
	\begin{figure}[!t]
		\centering
		\includegraphics[width=0.99\linewidth]{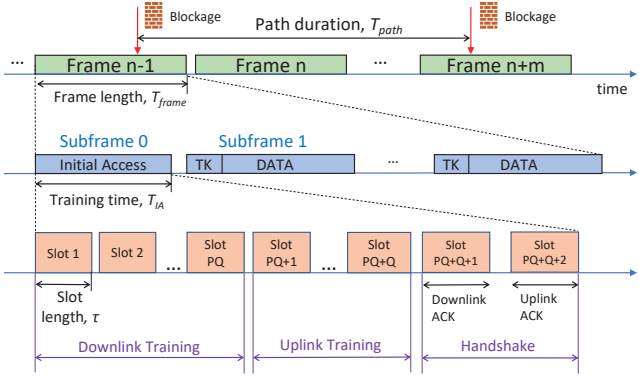}
		\caption{Frame structure. \textit{TK} indicates the tracking subframe.}
		\label{fig:frame}
	\end{figure}
	{Let time be partitioned into frames and consider a frame structure similar to the 5G new radio (NR) standard \cite{lien20175gnr} shown in Fig. \ref{fig:frame}. Each frame consists of multiple subframes, and each subframe consists of multiple time slots. A time slot is the minimum unit of time resource to be allocated, which consists of multiple symbols. There are two types of subframes: \emph{initial access} and \emph{standard}. The length of a standard subframe depends on the channel coherence time. The  length of the initial access subframe, together with the frame length and the slot duration, is optimized in Section \ref{sec:sys} to maximize the system throughput.}
	
	{The initial access subframe is used to establish links for newly scheduled mobiles or to recover links due to blockage, assuming no prior CSI. There is one initial access subframe within each frame, which is the first subframe. The standard subframes are used for data transmission. Each has a tracking period and a data transmission period. The tracking period is used to track small angular deviations of propagation paths and channel variations between subframes, and to refine the beamforming/combining vectors \cite{zhang2016tracking,zhang2016mobile,zhu2017abp,Palacios2017Tracking}. Since the channel paths persist across many coherence times \cite{niu2015survey,maccartney2017rapid}, the refinements are assumed to be minor and so the overhead for tracking is assumed negligible. We therefore focus on designing protocols for beam acquisition within the initial access subframe.}
	
	At the beginning of each frame, all APs and mobiles start an initial access procedure in which signals are transmitted in both downlink and uplink directions. The goal is to connect each mobile to an AP with a good channel. Mobiles that are not successfully connected, either due to bad channel realizations or limited system resources, will wait for the next frame and attempt to connect again. Successfully connected mobiles transmit and receive data during the standard subframes through the end of the frame. The multiple access protocol requires coarse frame synchronization, which means all the APs and mobiles are required to know the approximate beginning of a frame. 
	
	\subsection{Initial Access Protocol}
	{We assign each mobile a distinct narrow frequency band (slot) or unmodulated tone. Each mobile transmits and receives training signals only on its assigned
		narrow band. }
	{It is assumed that mobiles in a cell use distinct frequencies.}\footnote{{Those could be assigned via a control channel for cellular systems.}} 
	This narrowband design has the following advantages:

	{1) Mutual interference is eliminated during initial access. In practice, mobiles may be nearby (e.g., in a conference hall or a stadium), and their channels may share the same AoAs at an AP. Assigning different tones to different mobiles avoids pilot collisions and the APs can acknowledge their selected mobiles using their respective tones. This exploits the large bandwidth available in mmWave bands. 
	
	2) {The transmit and receive beam directions estimated on one narrow band are suitable for data transmission on other narrow bands as well. }
	Indeed, recent mmWave channel measurements \cite{rappaport2017overview} have shown that the directions of major propagation paths remain almost the same over
	{several GHz which spans multiple coherence bands}. Therefore, instead of probing a wide frequency band, it suffices to estimate the path directions on a narrow band. 
	
	3) Focusing training energy on a frequency slot boosts the SNR, which can reduce training error and overhead. 
	
	
	4) The complexity of narrowband signal processing is expected to be lower than that for wideband signaling. {Each mobile must only filter out signals outside its assigned frequency slot with a bandpass filter. }}
	
	%
	
	The initial access between an AP and a connecting mobile is illustrated in Fig. \ref{fig:procedure}. There are three stages: downlink training, uplink training, and handshake. This protocol is followed by all APs and mobiles. {We assume TDD, uplink/downlink reciprocity, and that the channels are constant during the training period. The simulation indicates that the total training time is typically around 3 ms, which is shorter than the observed coherence time \cite{va2015basic} and path \mbox{duration \cite{niu2015survey}.}}
	
	For concreteness, we describe the protocol with beam sweeping. The beamformers (or combiners) $\mathbf{f},\mathbf{g},\mathbf{w}$ take the form of $\mathbf{a}(\boldsymbol{\phi})$ or $\mathbf{u}(\boldsymbol{\theta})$, where all the sub-arrays are steered to the same direction. With a ULA phased array, signals or combiners can be directed to any desired azimuth angle by varying $\phi$ or $\theta$. With a UPA, both azimuth and elevation angles are changed to sweep over the 3-D space. {Note that the protocol applies to other training codebooks as well (see \cite{marzi2016compressive,hur2013millimeter,abari2016millimeter}). These variations may result in different training overhead, but do not require modifying the protocol.} 
	\begin{figure}[!t]
		\centering
		\includegraphics[width=0.95\linewidth]{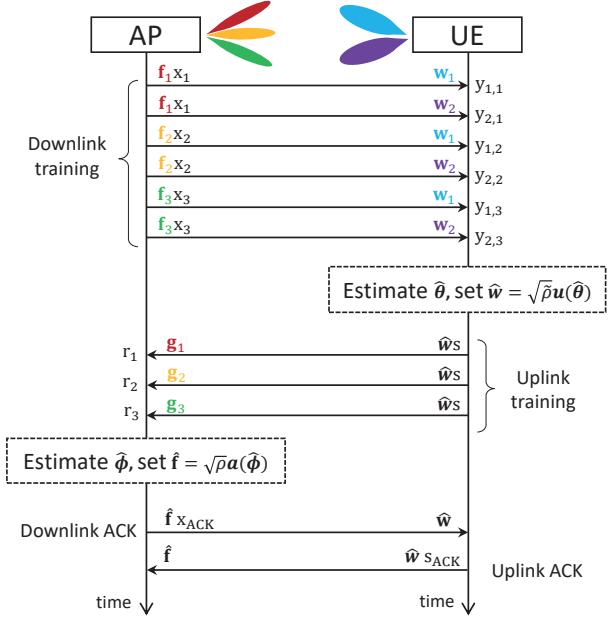}
		\caption{Example of the initial access procedure for $Q=3, P=2, I=1$. Only one AP and one mobile are shown, so their indices $l$ and $k$ are omitted.}
		\label{fig:procedure}
	\end{figure}
	
	The \emph{downlink training} spans $PQ$ time slots, where the APs sweep over $Q$ directions and the mobiles sweep over $P$ directions. Specifically, AP $l$ sequentially sends downlink pilots in $Q$ different directions using steering vectors $\mathbf{f}_{l,1},\mathbf{f}_{l,2},\cdots,\mathbf{f}_{l,Q}$, and mobile $k$ receives from $P$ directions with combiners $\mathbf{w}_{k,1},\mathbf{w}_{k,2},\cdots,\mathbf{w}_{k,P}$ in round-robin fashion. An AP uses the same beamforming vectors in all frequency bands, but mobile $k$ only detects the signal on its assigned frequency band $t_k$ (using a narrowband filter). In each time slot, corresponding to a particular steering vector $\mathbf{f}_{l,q}$ and combiner $\mathbf{w}_{k,p}$, the pilot symbol is repeated $I$ times, and the mobile averages the $I$ received samples to yield a sufficient statistic $y_{k,p,q}=\frac{1}{I}\sum_{i=1}^Iy_{k,p,q}[i]$. After $PQ$ time slots, mobile $k$ obtains $PQ$ samples $\{{y}_{k,p,q}\}$, which are used to estimate the direction of the strongest path $\hat{\boldsymbol{\theta}}_k$. The steering beam is then $\hat{\mathbf{w}}_k =\sqrt{\tilde{\rho}_k}\mathbf{u}(\hat{\boldsymbol{\theta}}_k) $.
	
	In the \emph{uplink training} stage, mobile $k$ uses $\hat{\mathbf{w}}_k$ as the beamformer and sends uplink signals over $Q$ time slots. Similarly, $I$ repeated pilots are sent within each time slot. Also, mobile $k$ only sends signals over its assigned frequency tone $t_k$. In the $q$-th time slot, AP $l$ combines signals at all frequency bands with the same combiner ${\mathbf{g}}_q$ and uses a bank of narrowband filters to separate the signals from different mobiles. The filtered baseband samples of \mbox{mobile $k$} are then averaged as $r_{k,l,q}$ which 
	does not contain interference from other mobiles 
	due to frequency orthogonality. After $Q$ time slots, AP $l$ estimates the direction of the strongest path from \mbox{mobile $k$} as $\hat{\boldsymbol{\phi}}_k$ based on samples $r_{k,l,1},r_{k,l,2},\cdots,r_{k,l,Q}$. Similarly, AP $l$ can estimate the direction of other mobiles and obtains the set of estimated angles $\hat{\boldsymbol{\phi}}_{l,1},\cdots,\hat{\boldsymbol{\phi}}_{l,K}$. 
	
	The \emph{handshake} has two time slots, one for downlink acknowledgment (ACK) and one for uplink ACK. Depending on the SNRs of mobiles and system constraints (e.g., traffic condition, number of available RF chains, and physical resources), AP $l$ schedules a subset of mobiles, and selects beamformer $\hat{\mathbf{f}}_{l,k}$ over different tones based on the estimated angles of those mobiles. In the case where only mobile $k$ is scheduled, $\hat{\mathbf{f}}_l = \sqrt{\rho_l}\mathbf{a}(\boldsymbol{\hat{\phi}}_{l,k};N)$. Then AP $l$ sends a downlink ACK message $x_\text{ACK}$ to mobile $k$ on frequency $t_k$. At the same time, mobile $k$ tries to detect downlink ACK messages with combiner $\hat{\mathbf{w}}_k$ on frequency $t_k$. Upon detecting the message, mobile $k$ responds to AP $l$ by sending an uplink ACK message $s_\text{ACK}$ on frequency band $t_k$ with \mbox{beamformer $\hat{\mathbf{w}}_k$.} Since both the APs and mobiles have estimated the appropriate beamforming/combining filters, the downlink/uplink ACK messages can be sent reliably and may contain additional information for establishing the data link. 
	
	
	
	
	\section{Channel Estimation}\label{sec:est-method}
	In this section, we present three methods for estimating the angles of the strongest path.
	{As previously mentioned, since the angles do not change over a wide range of frequencies, it suffices to perform the estimation on a narrow band. The estimated angles can be used to design beams for other narrow bands without causing performance loss.} We focus on a specific mobile's estimation problem with downlink training and drop the mobile index $k$. During downlink signaling, the AP explores $Q$ beams and the mobile explores $P$ beams. We assume that the training beam sequences $\mathbf{f}_1,\dots,\mathbf{f}_Q$ and $\mathbf{w}_1,\dots,\mathbf{w}_P$ have been specified according to some signaling protocol (e.g., sweeping, compressed sensing, etc). In the $(p,q)$-th time slot, AP $l$ repeats the pilot symbol $x_{l,q}$ for $I$ times to mitigate noise, and the mobile takes an average of these $I$ received samples. So, the downlink averaged
	received signal is
	\begin{align} \label{sig:est}
	\begin{split}
	y_{p,q} = \mathbf{w}_p^H \sum_{l=1}^L\sqrt{\rho_l}\mathbf{H}_{l}\mathbf{f}_{l,q} x_{l,q} + \mathbf{w}_p^H\frac{1}{I}\sum_{i=1}^I\mathbf{n}_{p,q}[i],
	\end{split}
	\end{align}
	where $\rho_l$ is the transmit power on a single tone,  $\mathbf{f}_{l,q}\in\mathbb{C}^{N}$ is the normalized beamforming vector at AP $l$ with $\|\mathbf{f}_{l,q}\|^2=1$, $x_{l,q}$ is the pilot symbol with $|x_{l,q}|^2=1$, and the noise \mbox{$\mathbf{n}_{p,q}[i]\sim \mathcal{CN}(0,\sigma_n^2\mathbf{I}_{M})$} is i.i.d. over $i,p,q$. We define an observation matrix $\mathbf{Y}\in\mathbb{C}^{P\times Q}$, with $(p,q)$-th element $y_{p,q}$. 
	
	\subsection{Maximum Power (MP)}
	The MP method chooses the beam pair $\left(\hat{p},\hat{q}\right)$ that yields the highest received power among the $PQ$ combinations, and takes the combining direction to be the $p$-th receive beam: \footnote[1]{For MP, we assume that directional beams are used for training.}
	\begin{align}\label{eq:mp}
	|y_{\hat{p},\hat{q}}|^2 \geq |y_{p,q}|^2,\text{for all }p\in\{1,\dots,P\} \text{ and }q\in\{1,\dots,Q\}.
	\end{align}
	MP has often appeared in previous work \cite{alkhateeb2015limited,zhao2017multiuser}, and is used in standards 
	(including IEEE 802.11ad). 
	Power detection is robust to phase errors and frequency offset. It is usually combined with beam sweeping or hierarchical search to exploit directional transmission. To achieve high estimation resolution, it requires searching a large beam space ($PQ$ beams). 
	MP in general needs many repeated pilots for each beam pair ($I>1$) to combat noise and fading. With limited training (fixed $IPQ$), there exists a tradeoff between the number of repeated pilots $I$ and the beam space size $PQ$. 
	
	\subsection{Maximum Likelihood (ML)} \label{sec:ml}
	ML methods compute the parameters that maximize the likelihood of observing the given signals. Here we make some simplifying assumptions about the channel model \eqref{eq:channel_withIndex}. Since the receiver determines a single beamforming direction, it is reasonable to assume that the received signals are transmitted from some AP $l$ through a \emph{single-path} channel with gain $\alpha$, AoA $\boldsymbol{\theta}$, and AoD $\boldsymbol{\phi}$. With i.i.d. noise $\tilde{n}_{p,q}\sim\mathcal{CN}(0,\sigma_n^2/I)$, the hypothesized received signal is then
	\begin{align} \label{eq:mismatch_model}
	\hat{y}_{p,q} =\alpha\mathbf{w}_p^H\mathbf{u}({\boldsymbol{\theta}})\mathbf{a}^H({\boldsymbol{\phi}})\mathbf{f}_{l,q} x_{l,q} + \tilde{n}_{p,q}.
	\end{align}  
	
	Conditioned on the training symbols and the parameters $(\boldsymbol{\theta},\boldsymbol{\phi},\alpha,l)$, the observed signals follow a multivariate normal distribution. Let $\mathbf{Z}(\boldsymbol{\theta},\boldsymbol{\phi},l)\in\mathbb{C}^{P\times Q}$ be a beamforming gain matrix with the $(p,q)$-the element defined as $z_{p,q}(\boldsymbol{\theta},\boldsymbol{\phi},l) = \mathbf{w}_p^H\mathbf{u}({\boldsymbol{\theta}})\mathbf{a}^H({\boldsymbol{\phi}}) \mathbf{f}_{l,q}$. With independent observations, the proposed ML method solves the problem:
	\begin{align} \label{ml:middle}
	\underset{\boldsymbol{\theta},\boldsymbol{\phi},\alpha,l}{\text{minimize}} \quad \|\alpha\mathbf{Z}(\boldsymbol{\theta},\boldsymbol{\phi},l) -\mathbf{Y}\|^2_F,
	\end{align}
	where $\|\cdot\|_F$ is the Frobenius norm. For fixed $\boldsymbol{\theta},\boldsymbol{\phi},l$, the optimal estimate of path gains \begin{align}\alpha^*({\boldsymbol{\theta}},{\boldsymbol{\phi}},l) = {\text{Tr}(\mathbf{Z}^H({\boldsymbol{\theta}},{\boldsymbol{\phi}},l)\mathbf{Y})}/{\|\mathbf{Z}({\boldsymbol{\theta}},{\boldsymbol{\phi}},l)\|_F^2}.
	\end{align}
	Then, we can rewrite \eqref{ml:middle} as
	\begin{align} \label{ml}
	\underset{\boldsymbol{\theta},\boldsymbol{\phi},l}{\text{maximize}} \quad  \frac{|\text{Tr}(\mathbf{Z}^H({\boldsymbol{\theta}},{\boldsymbol{\phi}},l)\mathbf{Y})|^2}{\|\mathbf{Z}({\boldsymbol{\theta}},{\boldsymbol{\phi}},l)\|_F^2}.
	\end{align}
	
	We note that the mismatch between the assumed model \eqref{eq:mismatch_model} and the original model \eqref{eq:channel_withIndex} is typically insignificant. In the case of multiple strong paths, the estimated AoA is the one that has the largest correlation with the received signals according to \eqref{ml}. 
	
	The least squares problem \eqref{ml} is nonlinear and challenging to solve. First, it requires a search over $l=1,2,\cdots,L$ to define the mapping $\mathbf{Z}({\boldsymbol{\theta}},{\boldsymbol{\phi}},l)$. Second, for a fixed $l$, the non-linear mapping $\mathbf{Z}({\boldsymbol{\theta}},{\boldsymbol{\phi}},l)$ is complicated, and the problem has a large number of local maxima. In \cite{zhang2016mobile}, the authors propose a solution method based on the Levenberg-Marquardt algorithm. It first uses MP to obtain an initial estimate, and then gradient descent to obtain a local optimum. However, calculation of the gradient requires a matrix inversion which is computationally expensive, and the performance largely depends on the initialization. Alternatively, in Section \ref{sec:FFT}, we present a solution method which uses FFTs to efficiently calculate \eqref{ml} and obtains near-optimal solutions with much lower computational complexity. 
	
	\subsection{Local Maximum Likelihood (LML)}
	To compute the ML estimate in Section \ref{sec:ml}, the receiver must know the transmitted beamforming vectors $\{\mathbf{f}_{l,q}\}$. We next describe the LML method, which assumes the transmitted beams are not available at the receiver. A similar approach is presented in \cite{marzi2016compressive}, where the received \mbox{matrix $\mathbf{Y}$} is sent back to the transmitter for AoD estimation. The feedback scheme \mbox{in \cite{marzi2016compressive}} cannot be directly applied here. This is because with multiple APs, we need to know to which AP to feed back. However, this is the outcome of the AP selection problem which requires the channel information being estimated.
	
	The LML method only estimates the AoA. First, consider the following {single-path}	model where the signal in the $(p,q)$-th slot is hypothesized to be transmitted through a single-path channel with gain $\beta_q$ and AoA $\boldsymbol{\theta}$:
	\begin{align} \label{eq:hch_lml}
	\hat{y}_{p,q} = \beta_q\mathbf{w}_p^H\mathbf{u}({\boldsymbol{\theta}}) x_{q} + \tilde{n}_{p,q},
	\end{align} 
	where $\beta_q=\alpha\mathbf{a}^H({\boldsymbol{\phi}})\mathbf{f}_{q}$ incorporates both path loss and beamforming gain during the $P$ time slots when the APs use the $q$-th precoders $\{\mathbf{f}_{l,q}\}_{1\leq l\leq L}$. For another period of $P$ time slots where the APs use the $q'$-th precoder, the received signals are hypothesized to be transmitted through another channel with a different gain $\beta_{q'}$ (due to the different precoder) but the same AoA $\boldsymbol{\theta}$. Conditioned on $\beta_1,\beta_2,\dots,\beta_Q$, and $\boldsymbol{\theta}$, the received signal is multivariate normal. 
	
	The LML method solves the following problem:
	\begin{align}\label{eq:lml_first}
	\underset{\boldsymbol{\theta},\beta_1,\beta_2,\dots,\beta_Q}{\text{maximize}}\quad  f_{\hat{y}_{1,1},\dots,\hat{y}_{P,Q}}(\mathbf{Y}|\boldsymbol{\theta},\beta_1,\beta_2,\dots,\beta_Q).
	\end{align}
	{There is a closed-form solution for $\beta_q$, which depends on $\boldsymbol{\theta}$, $\mathbf{w}$, and $\mathbf{Y}$. Substituting in \eqref{eq:lml_first}, we then wish to}
	\begin{align} \label{lml}
	\underset{\boldsymbol{\theta}}{\text{maximize}} \quad {\|\mathbf{b}^H({\boldsymbol{\theta}})\mathbf{Y}\|^2}/{\|\mathbf{b}({\boldsymbol{\theta}})\|^2},
	\end{align}
	where $\mathbf{b}(\boldsymbol{\theta})\in \mathbb{C}^{P}$ is a vector with the $p$-th element $b_p(\boldsymbol{\theta}) = \mathbf{w}_p^H\mathbf{u}({\boldsymbol{\theta}})$. Note that only the receiver's local combining vectors $\{\mathbf{w}_p\}$ are required to calculate \eqref{lml}. 
	
	\subsection{FFT calculation of decision statistic}\label{sec:FFT}
	In this section, we show that with uniform arrays (UPA or ULA) defined in Section \ref{sec:channelModel}, the decision statistics in \eqref{ml} and \eqref{lml} can be efficiently computed with FFTs. This is based on the observation that the antenna response functions for uniform arrays in \eqref{eq:upa_ula} are composed of DFT-type vectors. This method works for arbitrary training beams.
	
	For simplicity, we drop the AP index $l$, and write the numerator in \eqref{ml} as
	\begin{align}
	\begin{split} \label{eq:2dfft}
	\text{Tr}(\mathbf{Z}^H({\boldsymbol{\theta}},{\boldsymbol{\phi}})\mathbf{Y}) &= \sum_{p=1}^{P}\sum_{q=1}^{Q}  \mathbf{u}^H({\boldsymbol{\theta}})\mathbf{w}_p\mathbf{f}_q^H \mathbf{a}({\boldsymbol{\phi}})y_{p,q} \\
	&= \mathbf{u}^H({\boldsymbol{\theta}})\left({\sum_{p=1}^{P}\sum_{q=1}^{Q}  \mathbf{w}_p\mathbf{f}_q^Hy_{p,q}}\right) \mathbf{a}({\boldsymbol{\phi}}),
	\end{split}
	\end{align} 
	and the numerator in \eqref{lml} as
	\begin{align}\label{eq:1dfft}
	\|\mathbf{b}^H({\boldsymbol{\theta}})\mathbf{Y}\|^2 = \sum_{q=1}^Q \left| \sum_{p=1}^{P} \mathbf{u}^H({\boldsymbol{\theta}})\mathbf{w}_p y_{p,q} \right|^2= \sum_{q=1}^Q \left|\mathbf{u}^H({\boldsymbol{\theta}})\boldsymbol{\lambda}_q\right|^2,
	\end{align}
	where the $JM$-dimensional vector $\boldsymbol{\lambda}_q= \sum_{p=1}^{P}\mathbf{w}_p y_{p,q}$. 
	
	First, consider the case where ULA is used and the antenna spacing between sub-arrays is the same as the antenna spacing within a sub-array, that is, $u=Md$. Then the antenna response vector can be written as $\mathbf{u}(\theta) = \mathbf{e}(\vartheta;JM) $ with $\vartheta= 2\pi  d \sin(\theta)f_c/c$. Since the vector $\mathbf{e}(\vartheta;JM)$ is a DFT vector, each summation term $\mathbf{u}^H({\boldsymbol{\theta}})\boldsymbol{\lambda}_q$ in \eqref{eq:1dfft} is tantamount to a $JM$-point DFT of the vector $\boldsymbol{\lambda}_q$ evaluated at frequency $\vartheta$.  This motivates the use of FFT to reduce the computational complexity. By performing a $C$-point FFT on $\boldsymbol{\lambda}$, where $C$ is a power of two, we can jointly obtain the statistics at $C$ angles evenly dividing the full circle. For the case $JM<C$, we \mbox{pad $\boldsymbol{\lambda}_q$} with $C-JM$ zeros and perform a $C$-point FFT on the augmented vector. With sufficiently high quantization resolution $C$, this method guarantees a solution arbitrarily close to the global optimum. Similarly, the vectors $\mathbf{a}({\boldsymbol{\phi}})$ and $\mathbf{u}({\boldsymbol{\theta}})$ in \eqref{eq:2dfft} are also DFT vectors, and we can use a 2D-FFT to calculate the decision statistic. 
	
	Next, consider the general case where either ULAs or UPAs defined in \eqref{eq:upa_ula} are used. Let $\tilde{\mathbf{W}}_p\in\mathbb{C}^{M\times J}$ be a matrix taking every $M$ consecutive elements of $\mathbf{w}_p$ as a column, and let $\boldsymbol{\Lambda}_q = \sum_{p=1}^Py_{p,q}\tilde{\mathbf{W}}_p$. Using the fact that $\text{vec}(\mathbf{ABC}) = (\mathbf{C}^H\otimes\mathbf{A})\text{vec}(\mathbf{B})$,\footnote[1]{$\text{vec}([\mathbf{a}_1,\dots, \mathbf{a}_N]) = [\mathbf{a}_1^T,\mathbf{a}_2^T,\dots,\mathbf{a}_N^T]^T\in\mathbb{C}^{MN}$ for $\mathbf{a}_n\in\mathbb{C}^M, \forall n$.} we can rewrite each summation term $\mathbf{u}^H({\boldsymbol{\theta}})\boldsymbol{\lambda}_q$ in \eqref{eq:1dfft} as
	\begin{align} \label{eq:2dfft_general}
	\mathbf{u}^H(\boldsymbol{\theta})\boldsymbol{\lambda}_q &= \left(\mathbf{e}^H(\vartheta_1;J)\otimes \mathbf{e}^H(\vartheta_2;M)\right)\boldsymbol{\lambda}_q \\
	&= \mathbf{e}^H(\vartheta_2;M)  \boldsymbol{\Lambda}_q \mathbf{e}(\vartheta_1;J).
	\end{align} 
	Since $\mathbf{e}(\vartheta_1;J)$ and $\mathbf{e}(\vartheta_2;M)$ are DFT vectors, we can use a 2D-FFT to calculate \eqref{eq:2dfft_general}. Hence, \eqref{eq:1dfft} can be calculated with 2D-FFTs (azimuth and elevation AoAs) and \eqref{eq:2dfft} can be calculated with 4D-FFTs (azimuth and elevation AoAs and AoDs). 
	
	The denominators in \eqref{ml} and \eqref{lml} can be written as 
	\begin{align}
	\|\mathbf{Z}({\boldsymbol{\theta}},{\boldsymbol{\phi}})\|_F^2 = \sum_{p=1}^{P}\sum_{q=1}^{Q} |\mathbf{a}^H(\boldsymbol{\phi}) \mathbf{f}_q\mathbf{w}_p^H\mathbf{u}(\boldsymbol{\theta})|^2,
	\end{align}
	and
	\begin{align}
	\|\mathbf{b}(\boldsymbol{\theta})\|^2 = \sum_{p=1}^{P} |\mathbf{u}^H(\boldsymbol{\theta})\mathbf{w}_p|^2,
	\end{align}
	which can be calculated using FFTs as well. Since these are independent of the instantaneous observation $\mathbf{Y}$, each receiver can compute it offline. Note that if beam sweeping is implemented using a standard DFT codebook, then $\|\mathbf{Z}(\boldsymbol{\theta},\boldsymbol{\phi})\|_F^2$ and $\|\mathbf{b}(\boldsymbol{\theta})\|^2 $ are the same for all $(\boldsymbol{\theta},\boldsymbol{\phi})$ and $\boldsymbol{\theta}$, and can thus be removed from \eqref{ml} and \eqref{lml}. 
	
	The complexity of the FFT implementation is $O(C\log C)$, while direct calculation has complexity $O(CJM)$. The FFTs could be calculated with dedicated hardware modules \cite{son2002highspeed}.

	\subsection{Performance Analysis} \label{sec:ana}
	In this section, we present some insights into the MP and ML estimation. For simplicity and analytical tractability, we assume each AP (or mobile) uses ULA. We also assume beam sweeping for signaling, where both the beamforming and combining vectors are sampled from a DFT codebook. We take downlink signaling as an example and focus on a particular mobile. In contrast to \cite{bai2015coverage}, where an ideally sectored beam pattern is assumed, we consider a practical beam pattern having a main lobe and sidelobes. 
	
	Given a fixed total number of pilot symbols $\Omega=IPQ$, a question is whether to assign a different beam to each pilot, or to repeat pilots across a smaller set of beams. That is, the estimated directions are chosen from the $PQ$ combination of swept beams and increasing the number of sweep directions increases estimation resolution. On the other hand, repeating pilots $(I>1)$ for each beam direction suppresses noise. For MP, there is a non-trivial tradeoff between these two effects.
	
	{In contrast, we show that there is no such a tradeoff for ML in the following proposition}. 
	
	\begin{prop}
		{For ML estimation with ULA and beam sweeping, if the number of training beams satisfies $P\geq \bar{M}$ and $Q\geq \bar{N}$ where $\bar{M}$ and $\bar{N}$ are the number of antennas used for training, then the estimation error only depends on the total amount of training $\Omega=IPQ$. }
	\end{prop}
	\begin{proof}
		{With beam sweeping, if the AP uses $\bar{N}$ antennas for training, then it needs to sweep at least $\bar{N}$ directions to cover the whole space. Therefore, the sweep directions need to be $Q\geq \bar{N}$. Similarly, for a mobile, we need $P\geq \bar{M}$. Let $\boldsymbol{\psi}=(\boldsymbol{\theta},\boldsymbol{\phi},l)$ denote the parameters to be estimated and $\lambda(\boldsymbol{\psi}) = \omega\text{Tr}(\mathbf{Z}^H({\boldsymbol{\theta}},{\boldsymbol{\phi}},l)\mathbf{Y})$ denote the decision statistics in \eqref{ml}, where $\omega=\sigma_n\sqrt{I\bar{N}\bar{M}/PQ}$ is a constant normalizing the variance. In that case, the decision statistic $\lambda(\boldsymbol{\psi})$ satisfies}
		\begin{align}\label{eq:ml_decision_stats}
		\lambda(\boldsymbol{\psi}) \sim \mathcal{CN}\left(\sqrt{\frac{IPQ}{{\bar{N}\bar{M}}}}\sum_{l=1}^L\sum_{s=1}^S\sqrt{\gamma_{s,l}}G({\boldsymbol{\psi}},\boldsymbol{\psi}_{s,l}), 1 \right),
		\end{align}
		where $\gamma_{s,l} = \rho_l|\alpha_{s,l}|^2/\sigma_n^2$ is the received SNR when steering beams along the $s$-th path of the $l$-th AP, and  $G({\boldsymbol{\psi}},\boldsymbol{\psi}_{s,l}) = \mathbf{u}^H({\boldsymbol{\theta}})\mathbf{u}({\boldsymbol{\theta}_{s,l}})\mathbf{a}^H({\boldsymbol{\phi}_{s,l}})\mathbf{a}({\boldsymbol{\phi}})$ characterizes the beamforming gain. 
		
		For ML, the estimate $\hat{\boldsymbol{\psi}}$ is the $\boldsymbol{\psi}$ that maximizes $|\lambda(\boldsymbol{\psi})|^2$, so the estimation error is uniquely determined by $\lambda(\boldsymbol{\psi})$. {The dependence on $I,P,Q$ enters \eqref{eq:ml_decision_stats} only through the product $\Omega=IPQ$, hence the probability of selecting an incorrect beam pair depends only on this total amount of training $\Omega$.} 
	\end{proof}
	
	\begin{figure}[!t]
		\centering
		\begin{subfigure}{.95\linewidth}
			\centering
			\includegraphics[width=\linewidth]{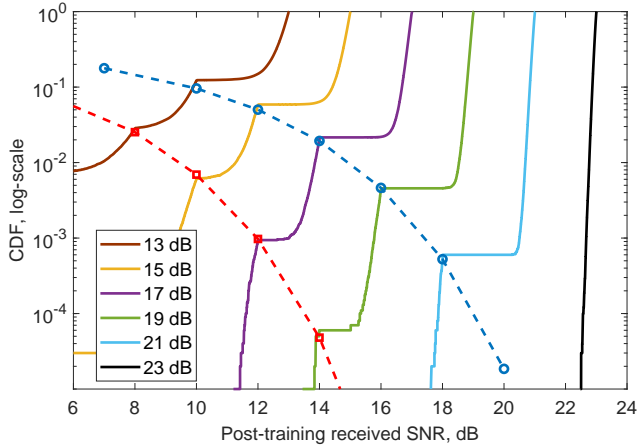}
			\label{fig:mpath}
		\end{subfigure}
		\caption{Received SNR of ML estimation after training. Dashed lines show the normal approximation for the probability of aligning with the second path (blue) and third path (red).}
		\label{fig:linkAna_simu}
	\end{figure}
	Next, we show the performance of ML with multiple paths. 
	We simulate a point-to-point three-path channel where the second and third paths are 3 dB and 5 dB weaker than the main path, respectively. Fig. \ref{fig:linkAna_simu} shows the distribution of post-training SNR, which is the received SNR when steering beams to estimate directions $\hat{\gamma} =| \mathbf{u}^H(\hat{\boldsymbol{\theta}})\mathbf{H}\mathbf{a}(\hat{\boldsymbol{\phi}})  |^2$. \footnote[1]{{Note that the error in the $\mathbf{u}$-$\mathbf{v}$ space $\epsilon_\text{uv}(\theta) = |\sin{(\hat{\theta})}-\sin{(\theta)}|$ does not depend on the path direction $\theta$, and so neither does the post-training SNR $\hat{\gamma}$. However, the error in the spatial domain $\epsilon_\text{spatial}(\theta) = |\hat{\theta}-\theta|$ does change with $\theta$, and the largest error occurs at $\theta=\pm \pi/2$.}} The training time is $IPQ=NMJ^2$. Results for different SNRs are presented. 
	
	{Fig. \ref{fig:linkAna_simu} shows that at high SNRs the distribution is approximately exponential. At lower SNRs, the distribution is a piece-wise function. This is because the estimated directions are mis-matched around the second or third path, instead of the strongest path. We can approximate the probability of choosing the direction around the $s$-th path as $\Pr(\hat{\boldsymbol{\psi}}=\boldsymbol{\psi}_s) \approx \Pr(\lambda(\boldsymbol{\psi_s})>\lambda(\boldsymbol{\psi_1})) \approx Q\left(\frac{\sqrt{\gamma_\text{max}}-\sqrt{\gamma_{s}} }{\sqrt{NMJ^2/IPQ}}\right)$, where $\gamma_\text{max}$ is the received SNR when aligned to the strongest path and $Q(\cdot)$ is the standard normal cumulative density function (CDF). In Fig. \ref{fig:linkAna_simu}, this approximation is shown with dashed lines. Even for a single-path channel, when the training SNR is low, sidelobes can produce a similar effect as secondary paths. To achieve a target error probability, the training SNR should be sufficiently high to mitigate the effect of secondary paths and sidelobes.} 
	
	\section{How much training is needed?} \label{sec:sys}
	In this section, we consider the problem of maximizing the system throughput by optimizing the key parameters of the protocol presented in Section \ref{sec:access-protocol}. Recall the frame structure in Fig. \ref{fig:frame}, where we assume the overhead due to handshaking and beam tracking is negligible. The optimization variables include the frame length $T_\text{frame}$, initial access duration $T_\text{IA}$, and downlink/uplink pilot slot duration $T_\text{slot}$ in the initial access subframe. 
	
	\subsection{Blockage Model}
	Due to the movement of a mobile and surrounding objects, its transmission path could be frequently blocked \cite{maccartney2017rapid}. When this occurs, initial access is required for discovering another path and re-establishing a connection. We consider a two-state Markov blockage model as in \cite{maccartney2017rapid}, where the probability of blockage is $\delta$. Since blockages are usually caused by nearby pedestrians or other objects, which can be modeled as a Poisson process, we assume the duration of a \mbox{path $T_\text{path}$} is an exponential random variable with \mbox{mean $1/\delta$}. We further define the data transmission time as 
	$T_\text{data} = \max\{\min\left\{T_\text{path}, T_\text{frame}\right\}  - T_\text{IA},0\},$
	which is a non-negative random variable with expectation 
	\begin{align}
	\mathbb{E}[T_\text{data}] = ({e^{-\delta T_\text{IA}} - e^{-\delta T_\text{frame}}})/{\delta}.
	\end{align}
	
	\subsection{Throughput Optimization}
	We consider the problem of maximizing the long-term throughput with respect to $T_\text{IA}$, $T_\text{frame}$, and $T_\text{slot}$. A longer training time $T_\text{IA}$ reduces the training error and increases the data rate; however, this increases overhead. On the other hand, a longer frame length $T_\text{frame}$ reduces training overhead, but a transmission is more likely to be blocked within a frame, leaving the rest of the frame empty. So there exists a design tradeoff for those parameters. 
	
	We assume each time slot in the initial access subframe has a single training symbol with bandwidth $B_\text{tr}$. Adjacent slots are separated by a guard interval of length $\tau$, so the slot duration is $T_\text{slot}=\tau+1/B_\text{tr}$. Assuming beam sweeping and ML estimation, the optimization problem is
	\begin{subequations}
		\begin{flalign}
		\underset{T_\text{IA}, T_\text{frame}, B_\text{tr} \geq 0}{\text{maximize}} \quad &\mathbb{E}_{\delta,\mathbf{n},\mathbf{h}}\left[\frac{T_\text{data}}{T_\text{frame}} \log(1+\gamma) \right], \\
		\text{subject to} \quad & \tau+1/B_\text{tr} \geq T_\text{switch} \label{c1},\\
		& T_\text{IA}B_\text{tr} \geq NMJ^2 \label{c2},\\
		& T_\text{IA} \leq T_\text{frame} \leq T_\text{max}  \label{c3}, 
		\end{flalign}
	\end{subequations}
	where {$\gamma$ is the SINR for data transmission using the estimated beams,} $T_\text{switch}$ is the minimum beam switching time due to hardware implementation of phase shifters \cite{sadhu2017phasearray}, and $T_\text{max}$ is the maximum frame length given by latency requirements. The expectation is over training error (caused by noise $\mathbf{n}$), random blockage $\delta$, and channel realization $\mathbf{h}$. To cover all spatial directions with sweeping, \eqref{c2} constrains the number of training beams to be at least the number of antennas. For very-large antenna arrays (hundreds of antennas), we propose to use a subset of antennas for training, but all antennas for data transmission. However, if compressive sensing (random beamforming) is used for signaling, then \mbox{constraint \eqref{c2}} can be removed. 
	
	The received {SINR} depends on $T_\text{IA}$ and $B_\text{tr}$ and is independent of both random blockage $\delta$ and $T_\text{frame}$. Also, $T_\text{data}$ is independent of the training error. So, the expectation can be decoupled and the problem becomes
	\begin{subequations}\label{problem:ovhd}
		\begin{flalign}
		\underset{T_\text{IA}, T_\text{frame}, B_\text{tr}\geq 0}{\text{maximize}}\ &\frac{e^{-\delta T_\text{IA}} - e^{-\delta T_\text{frame}}}{\delta T_\text{frame}}  \int_{0}^{\infty} \log(1+ x)  \ dF_X(x), \\
		\text{subject to} \quad & \eqref{c1}, \eqref{c2}, \eqref{c3},
		\end{flalign}
	\end{subequations}
	where {$F_X(x)$ is the CDF of the data transmission SINR}. Since obtaining an analytical expression for $F_X(x)$ is difficult, we propose to evaluate it through the Monte Carlo method. 
	
	
	{To solve \eqref{problem:ovhd}, we first observe that the optimal value for $B_\text{tr}$ is $1/(T_\text{switch}-\tau)$ for all $T_\text{IA},T_\text{frame}$. The reason is that the integral in the objective only depends on the product $T_\text{IA}B_\text{tr}$, and the first term in the objective increases when $T_\text{IA}$ decreases. Hence, for fixed $T_\text{IA}B_\text{tr}$, we should make $B_\text{tr}$ as large as possible, which is the upper bound $1/(T_\text{switch}-\tau)$. Next, with $T_\text{IA}$ fixed, the objective is a concave function of $T_\text{frame}$. Although there is no closed-form solution, we can solve for $T_\text{frame}$ numerically with gradient descent. Finally, because $T_\text{IA}$ is determined by the number of pilots transmitted, it is a discrete variable in a finite set and can be optimized by exhaustive search. For the examples considered, we observe that solving $T_\text{IA}$ with gradient descent also gives the optimal solution. This is because the objective appears to be concave over $T_\text{IA}$, so the problem is quasi-concave. However, concavity over $T_\text{IA}$ cannot be proved because there is no explicit expression for SINR as a function of $T_\text{IA}$.}
	
	\section{Performance Evaluation} \label{sec:simu}
	
	\subsection{Training codebooks}
	We first compare the link-level performance of different training codebooks. We consider a network with 3 APs and 100 mobile devices. The three APs are arranged in a triangle with inter-AP distance 250 m. The mobiles are randomly dropped within the polyhedron with the minimum distance to an AP of 15 m. We use the 3GPP Urban Micro (UMi) path loss model with the carrier frequency of 28 GHz \cite{3gpp-38-900}. We assume each AP and mobile is equipped with a ULA with $N=M=16$ antennas and $J=2$ sub-arrays. The distance between adjacent sub-arrays is the same as the antenna element spacing within a sub-array, which is half of the carrier wavelength. The bandwidth for each narrow band training signal is 250 kHz, the minimum beam switching time is 4 $\mu s$, and the slot length is 8 $\mu s$. The training powers of APs and mobiles are 20 dBm and 15 dBm, respectively. 
	
	\begin{figure}[!t]
		\centering
		\includegraphics[width=.95\linewidth]{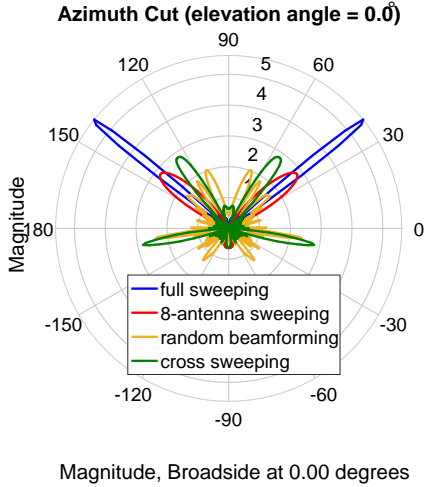}
		\caption{Example beam patterns with different signaling methods for a ULA with 16 antennas. }
		\label{fig:beamPattern}
	\end{figure}
	We simulate the following 
	{training codebooks} which differ in the type of beamforming/combining vectors $\mathbf{f},\mathbf{w},\mathbf{g}$ used in downlink and uplink signaling. With \emph{full sweeping}, the training signals are sent/received with DFT beams using all the antennas. With \emph{single-RF sweeping}, only one sub-array is activated for training, and as a result, the beams are wider with a DFT codebook. With \emph{adaptive sweeping}, the number of activated antennas is proportional to the number of search directions. For example, with $Q$ search directions at an AP, the first $\min(Q,NJ)$ antennas are activated. \emph{Cross sweeping} \cite{abari2016millimeter} is an alternative design of wide beams. The first half and second half of the antennas point to two orthogonal directions. \emph{Random beamforming} \cite{marzi2016compressive} is motivated by compressive sensing. The phase of each phase shifter is chosen randomly, and the resulting beam is omni-directional with random gains. For all of these schemes, the total transmission power is the same, and is equally split over all active antennas. 
	
	Fig. \ref{fig:beamPattern} shows an example of beam patterns corresponding to the different 
	{training codebooks}. It shows the magnitude of the inner product of the precoder $\mathbf{f}$ and an antennas response vector $\mathbf{a}(\theta)$ for $\theta$ in $[0,2\pi]$. The plot is symmetric about $\pm\pi/2$ because $\mathbf{a}(\theta)$ is determined by $\sin(\theta)$ instead of $\theta$ as in \eqref{eq:ula}. 
	
	\begin{figure}[!t]
		\centering
		\begin{subfigure}{.475\textwidth}
			\centering
			\includegraphics[width=\linewidth]{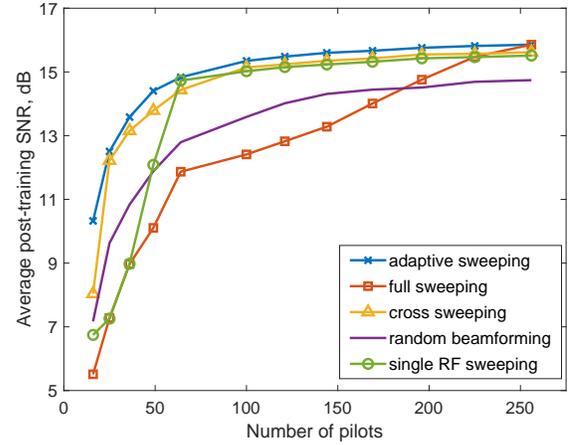}
			\caption{16 antennas at an AP.}
			\label{fig:16ant}
		\end{subfigure}
		\begin{subfigure}{.475\textwidth}
			\centering
			\includegraphics[width=\linewidth]{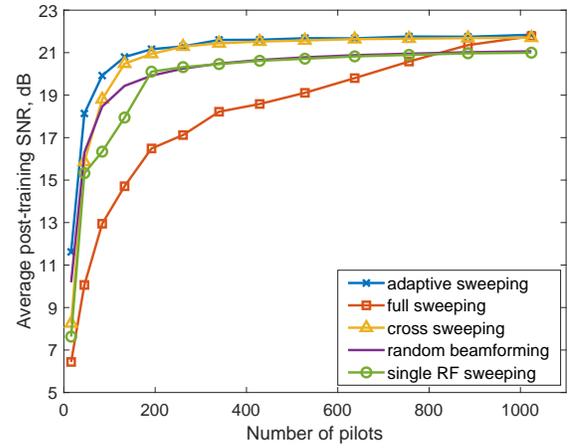}
			\caption{{64 antennas at an AP.}}
			\label{fig:64ant}
		\end{subfigure}
		\caption{Comparison of post-training SNR versus number of pilots for different beam sweeping methods. There are 16 antennas at a mobile.}
	\end{figure}
	
	The performance of different codebooks are shown in Fig. \ref{fig:16ant} and Fig. \ref{fig:64ant}. We focus on a typical mobile and use the ML method for channel estimation with an FFT size \mbox{of 64}. The post-training SNR is obtained by steering beams at both the AP and the mobile towards the estimated beamforming direction, using \emph{all} antennas, i.e., $\hat{\gamma} =| \mathbf{u}^H(\hat{\boldsymbol{\theta}})\mathbf{H}\mathbf{a}(\hat{\boldsymbol{\phi}})  |^2$. {There are $16$ antennas at a mobile. We show two examples with $16$ and $64$ antennas at an AP.} 
	
	The simulation results indicate that, among the codebooks considered, \textit{adaptive sweeping} performs best regardless of training time or antenna array size. Random beamforming generally needs more antennas and training time to achieve comparable performance. In either scenario, sweeping with all antennas does not perform well because with limited training, the narrow beams cannot cover all the spatial directions. By comparison, employing wider beams (either through single-RF sweeping or cross sweeping) improves performance when the training is limited. {Increasing the number of antennas at the APs improves the relative performance of random beamforming but does not significantly affect the other schemes. Also, the training overhead with adaptive sweeping increases only slightly with more antennas.}
	
	\subsection{Estimation Methods}
	\begin{figure}[!t]
		\centering
		\begin{subfigure}{.475\textwidth}
			\centering
			\includegraphics[width=\linewidth]{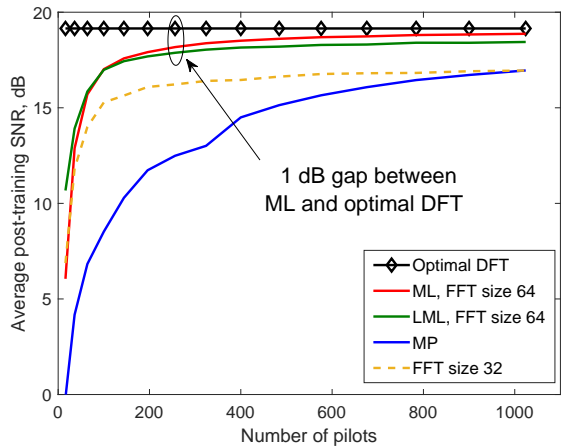}
			\caption{Power budget at each AP is 17 dBm.}
			\label{fig:bfGain_22}
		\end{subfigure}
		\begin{subfigure}{.475\textwidth}
			\centering
			\includegraphics[width=\linewidth]{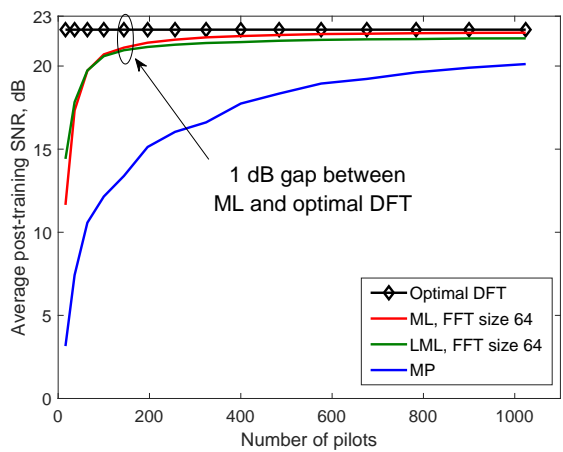}
			\caption{Power budget at each AP is 20 dBm.}
			\label{fig:bfGain_25}
		\end{subfigure}
		\caption{Comparison of post-training SNR for different estimation methods. }
	\end{figure}
	{Next, we compare the performance of different channel estimation methods discussed in Section \ref{sec:est-method} with the same training codebook. Adaptive beam sweeping is used and the number of search directions (codebook size) at an AP/mobile is $\sqrt{\Omega}$, $\Omega$ being the number of pilots.} We simulate two scenarios with AP power budgets 17 dBm and 20 dBm. The results are shown in Figs. \ref{fig:bfGain_22} and \ref{fig:bfGain_25}. The \textit{optimal DFT} algorithm takes a 2-D DFT of the channel matrix and selects the angles (AoA and AoD) with the largest magnitude. This is the maximum received SNR that can be obtained using beam steering. The ML and LML methods both perform uniformly better than the MP. With increasing number of training pilots, the performance of both methods approaches the global optimum, whereas the MP method has a performance loss due to the quantization of training beams. {Fig. \ref{fig:bfGain_22} includes ML with an FFT size of 32 and shows that ML can achieve the same estimation accuracy of MP with much less training.} There is little gap between the ML and LML results, so the analysis of ML in Section \ref{sec:ana} also gives an accurate estimate of the performance of LML. Comparing the ML curves in Figs. \ref{fig:bfGain_22} and \ref{fig:bfGain_25}, to achieve an SNR within $1$ dB of the upper bound, the required training time in Fig. \ref{fig:bfGain_22} is about twice that shown in Fig. \ref{fig:bfGain_25}. Since the power difference between the two figures is 3 dB, this is consistent with the decision statistic in \eqref{eq:ml_decision_stats} where doubling the training time effectively doubles the power. 
	
	\begin{figure}[!t]
		\centering
		\includegraphics[width=.95\linewidth]{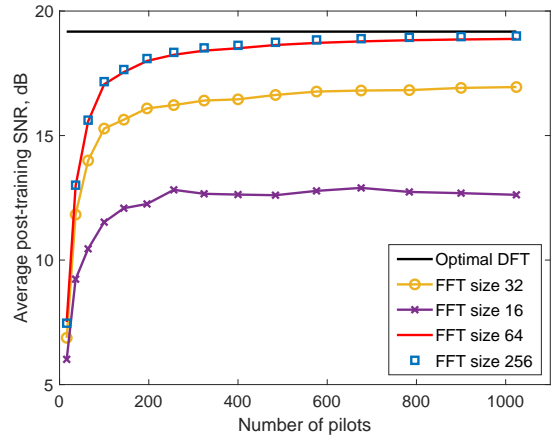}
		\caption{Averaged post-training SNR for the ML method with different FFT sizes.}
		\label{fig:fftSize}
	\end{figure}
	Fig. \ref{fig:fftSize} compares the performance of the ML method with varying FFT sizes. The results show that near-optimal performance can be obtained with a moderate FFT size of $64$. If the antenna array is very large and the beams used for data transmission are very narrow, a larger FFT size might be desired to increase resolution. 
	
	\subsection{Campus Scenario}
	In this section, we show results corresponding to a scenario in which mmWave APs are deployed on a college campus. Geographic information about buildings and roads for the Evanston campus of Northwestern University (Evanston, Illinois) are obtained from OpenStreetMap \cite{OpenStreetMap}. The map and the abstraction are shown in Fig. \ref{fig:sim_map}. We generate the urban micro (UMi) scenario in NYUSIM \cite{sun2017nyusim} with default environmental parameters. We place 10 APs in a hexagonal topology with inter-AP distance of approximately 200 m. The APs are assumed to be on the top of buildings with antenna height of 10 m. Mobiles are uniformly distributed on the roads with moving speed of 3 km/h. The antenna height at a mobile is 1.5 m. {Each AP has 32 ULA antennas and each mobile has 16 ULA antennas.}
	
	At mmWave frequencies, propagation paths can be easily blocked by trees or other pedestrians. We simulate those effects by randomly placing 2,000 small blockages in the system with size \mbox{1 m$^2$.} 
	Based on actual geographical locations of APs and mobiles, the channels are line-of-sight (LoS) if there is no blockage (buildings or small obstacles) between the transmitter and receiver; and otherwise are non-line-of-sight (NLoS). There is a LoS path and multiple NLoS paths in a LoS channel; the NLoS channels only contains multiple NLoS paths. The AoA, AoD, and delay of a LoS path are calculated based on the geographic locations of the AP and mobile. For all NLoS paths (of both LoS channels and NLoS channels), we assume their AoAs, AoDs, and phase delay are uniformly distributed in $[0,2\pi]$ for simplicity. 
	
	\begin{figure}[!t]
		\centering
		\includegraphics[width=.9\linewidth]{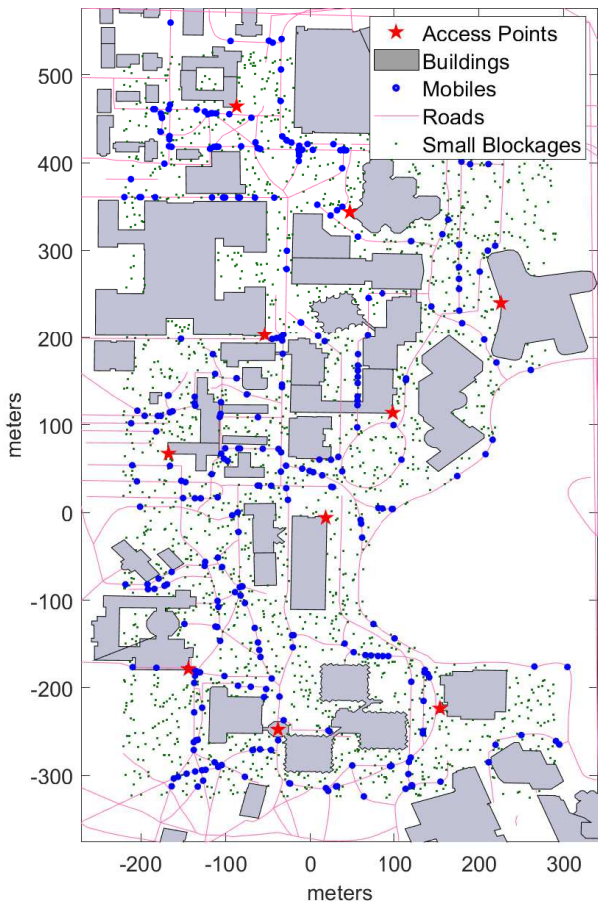}
		\caption{Campus simulation map showing AP locations and obstructions.}
		\label{fig:sim_map}
	\end{figure}
	
	For beam training, we use adaptive sweeping with LML estimation. The data transmission uses a total bandwidth of 100 MHz, which is further divided into 10 sub-bands with 10 MHz each. Beam steering is used for data transmission where the coefficients of beamforming and combining vectors are adjusted to the transmission frequency using estimated angles. Frequency division multiplexing (FDM) controls inter-user interference during data transmission. Specifically, mobiles served by the same AP are sorted according to their estimated AoDs and are assigned frequency slots in round-robin fashion. Mobiles with similar AoDs are assigned to different frequencies to reduce mutual interference. At the receiver, the maximum data receiving SINR is capped at 30 dB (in part due to quantization errors). 
	\begin{figure}[!t]
		\centering
		\begin{subfigure}{.475\textwidth}
			\centering
			\includegraphics[width=\linewidth]{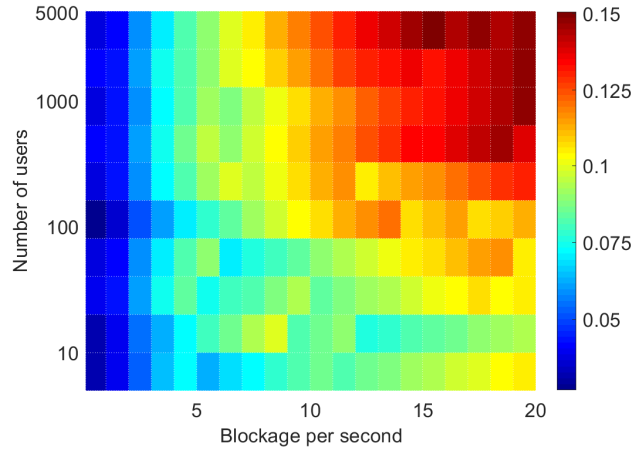}
			\caption{Maximum frame length 100 ms.}
			\label{fig:time}
		\end{subfigure}
		\begin{subfigure}{.475\textwidth}
			\centering
			\includegraphics[width=\linewidth]{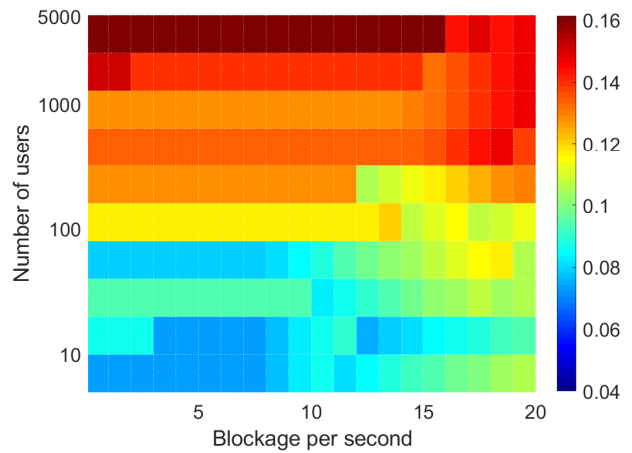}
			\caption{Maximum frame length 20 ms.}
			\label{fig:tr}
		\end{subfigure}
		\caption{Optimal training overhead.}
		\label{fig:detail}
	\end{figure}
	
	We compute the optimal training overhead for the whole system by solving problem \eqref{problem:ovhd}. Fig. \ref{fig:detail} shows how the optimized training overhead varies with the number of mobiles and the blocking rate. Since the time until an LoS path is blocked is typically more than a hundred milliseconds, which is larger than the frame length, we only consider blockage of NLoS paths. Different colors indicate different overhead levels listed in the colorbar. The Fig. \ref{fig:time} shows results for the scenario with maximum frame length 100 ms. With a moderate number of mobiles and blocking rate, the training overhead is around $5\%$; in extreme cases with very large number of mobiles and high blocking rate, it could exceed $10 \%$. In Fig. \ref{fig:tr}, we show the results with maximum frame length of $20$ ms. The training overhead, in this case, is similar to the 100-ms case when the blocking rate is high, whereas the overhead is substantially higher than the 100-ms case when the blocking rate is low. This is because the optimal frame length at low-blockage scenarios reaches the 20 ms constraint, so the training is initiated more often than necessary. A simple modification to address this issue is to let only mobiles that are blocked in the previous frame join the initial access process. This blockage occurrence can be readily detected in beam tracking phases. For mobiles that are not blocked, the APs can continue to transmit data with previously estimated beamformers. 
	
	\section{Conclusions} \label{sec:conclusion}
	In this paper, we have investigated the design and analysis of a mmWave network consisting of multiple APs and mobiles. We have proposed a narrowband training protocol along with different codebooks and estimation methods for beam acquisition. Simulation results indicate that adaptive sweeping with LML estimation achieves the best performance with moderate complexity. Campus scenario simulation shows that the training overhead with the proposed scheme is typically around $5\%$ and may exceed $10\%$ in a high-mobility environment or with high network loads.

	\textit{\bibliographystyle{ieeetr}}
	\bibliography{IEEEabrv,BeamAcquisition_JSAC}

\begin{thebibliography}{10}

\bibitem{hao2018initialaccess}
H.~{Zhou}, D.~{Guo}, and M.~L. {Honig}, ``Initial access and beamforming in
  multi-cell {mmWave} networks using narrowband pilots,'' in {\em Proc.
  ASILOMAR}, pp.~582--586, Oct. 2018.

\bibitem{wong2017key}
V.~W. Wong, {\em Key technologies for {5G} wireless systems}.
\newblock Cambridge university press, 2017.

\bibitem{rappaport2017overview}
T.~S. {Rappaport}, Y.~{Xing}, G.~R. {MacCartney}, A.~F. {Molisch},
  E.~{Mellios}, and J.~{Zhang}, ``Overview of millimeter wave communications
  for fifth-generation (5{G}) wireless networks---with a focus on propagation
  models,'' {\em IEEE Trans. Antennas Propag.}, vol.~65, pp.~6213--6230, Dec.
  2017.

\bibitem{desai2014initial}
V.~Desai, L.~Krzymien, P.~Sartori, W.~Xiao, A.~Soong, and A.~Alkhateeb,
  ``Initial beamforming for {mmWave} communications,'' in {\em Proc. ASILOMAR},
  pp.~1926--1930, Nov. 2014.

\bibitem{barati2016initial}
C.~N. Barati, S.~A. Hosseini, M.~Mezzavilla, T.~Korakis, S.~S. Panwar,
  S.~Rangan, and M.~Zorzi, ``Initial access in millimeter wave cellular
  systems,'' {\em IEEE Trans. Wireless Commun.}, vol.~15, pp.~7926--7940, Dec.
  2016.

\bibitem{giordani2016initial}
M.~{Giordani}, M.~{Mezzavilla}, and M.~{Zorzi}, ``Initial access in {5G mmWave}
  cellular networks,'' {\em IEEE Commun. Mag.}, vol.~54, pp.~40--47, Nov. 2016.

\bibitem{zhang2016tracking}
C.~Zhang, D.~Guo, and P.~Fan, ``Tracking angles of departure and arrival in a
  mobile millimeter wave channel,'' in {\em Proc. IEEE ICC}, pp.~1--6, May
  2016.

\bibitem{zhang2016mobile}
C.~Zhang, D.~Guo, and P.~Fan, ``Mobile millimeter wave channel acquisition,
  tracking, and abrupt change detection,'' {\em arXiv preprint
  arXiv:1610.09626}, 2016.

\bibitem{zhu2017abp}
D.~Zhu, J.~Choi, and R.~W. Heath, ``Auxiliary beam pair enabled {A}o{D} and
  {A}o{A} estimation in closed-loop large-scale millimeter-wave {MIMO}
  systems,'' {\em IEEE Trans. Wireless Commun.}, vol.~16, pp.~4770--4785, July
  2017.

\bibitem{Palacios2017Tracking}
J.~{Palacios}, D.~{De Donno}, and J.~{Widmer}, ``Tracking mm-{Wave} channel
  dynamics: Fast beam training strategies under mobility,'' in {\em IEEE
  INFOCOM 2017 - IEEE Conference on Computer Communications}, pp.~1--9, May
  2017.

\bibitem{Simic2017coverage}
L.~{Simic}, S.~{Panda}, J.~{Riihijarvi}, and P.~{Mahonen}, ``Coverage and
  robustness of {mm-Wave} urban cellular networks: Multi-frequency {HetNets}
  are the {5G} future,'' in {\em Proc. IEEE SECON}, pp.~1--9, June 2017.

\bibitem{gonzalez2018channel}
J.~P. Gonz{\'a}lez-Coma, J.~Rodriguez-Fernandez, N.~Gonz{\'a}lez-Prelcic,
  L.~Castedo, and R.~W. Heath, ``Channel estimation and hybrid precoding for
  frequency selective multiuser {mmWave} {MIMO} systems,'' {\em IEEE J. Sel.
  Topics Signal Process.}, vol.~12, pp.~353--367, May 2018.

\bibitem{sun2019beam}
X.~Sun, C.~Qi, and G.~Y. Li, ``Beam training and allocation for multiuser
  millimeter wave massive {MIMO} systems,'' {\em IEEE Trans. Wireless Commun.},
  vol.~18, pp.~1041--1053, Feb. 2019.

\bibitem{alkhateeb2015limited}
A.~Alkhateeb, G.~Leus, and {R. W. Heath Jr.}, ``Limited feedback hybrid
  precoding for multi-user millimeter wave systems,'' {\em IEEE Trans. Wireless
  Commun.}, vol.~14, pp.~6481--6494, Nov. 2015.

\bibitem{hur2013millimeter}
S.~Hur, T.~Kim, D.~J. Love, J.~V. Krogmeier, T.~A. Thomas, and A.~Ghosh,
  ``Millimeter wave beamforming for wireless backhaul and access in small cell
  networks,'' {\em IEEE Trans. Commun.}, vol.~61, pp.~4391--4403, Oct. 2013.

\bibitem{zhao2017multiuser}
L.~Zhao, D.~W.~K. Ng, and J.~Yuan, ``Multi-user precoding and channel
  estimation for hybrid millimeter wave systems,'' {\em IEEE J. Sel. Areas
  Commun.}, vol.~35, pp.~1576--1590, July 2017.

\bibitem{Sun2018Hybridbeamforming}
S.~{Sun}, T.~S. {Rappaport}, and M.~{Shaft}, ``Hybrid beamforming for 5g
  millimeter-wave multi-cell networks,'' in {\em IEEE INFOCOM 2018 - IEEE
  Conference on Computer Communications Workshops (INFOCOM WKSHPS)},
  pp.~589--596, April 2018.

\bibitem{zhao2017tone}
L.~Zhao, G.~Geraci, T.~Yang, D.~W.~K. Ng, and J.~Yuan, ``A tone-based {AoA}
  estimation and multiuser precoding for millimeter wave massive {MIMO},'' {\em
  IEEE Trans. Wireless Commun.}, vol.~65, pp.~5209--5225, Dec. 2017.

\bibitem{marzi2016compressive}
Z.~Marzi, D.~Ramasamy, and U.~Madhow, ``Compressive channel estimation and
  tracking for large arrays in mm-wave picocells,'' {\em IEEE J. Sel. Topics
  Signal Process.}, vol.~10, pp.~514--527, April 2016.

\bibitem{Ghadikolaei2015mmWave}
H.~{Shokri-Ghadikolaei}, C.~{Fischione}, G.~{Fodor}, P.~{Popovski}, and
  M.~{Zorzi}, ``Millimeter wave cellular networks: A {MAC} layer perspective,''
  {\em IEEE Trans. Wireless Commun.}, vol.~63, pp.~3437--3458, Oct. 2015.

\bibitem{li2017design}
Y.~Li, J.~G. Andrews, F.~Baccelli, T.~D. Novlan, and C.~J. Zhang, ``Design and
  analysis of initial access in millimeter wave cellular networks,'' {\em IEEE
  Trans. Wireless Commun.}, vol.~16, pp.~6409--6425, Oct. 2017.

\bibitem{li2017onthebeamformed}
Y.~Li, J.~Luo, M.~H.~C. Garcia, R.~B{\"o}hnke, R.~A. Stirling-Gallacher, W.~Xu,
  and G.~Caire, ``On the beamformed broadcasting for millimeter wave cell
  discovery: Performance analysis and design insight,'' {\em IEEE Transactions
  on Wireless Communications}, vol.~17, pp.~7620--7634, Nov. 2018.

\bibitem{bai2015coverage}
T.~Bai and {R. W. Heath Jr.}, ``Coverage and rate analysis for millimeter-wave
  cellular networks,'' {\em IEEE Trans. Commun.}, vol.~14, pp.~1100--1114, Feb.
  2015.

\bibitem{zhao2018multi}
L.~Zhao, Z.~Wei, D.~W.~K. Ng, J.~Yuan, and M.~C. Reed, ``Multi-cell hybrid
  millimeter wave systems: Pilot contamination and interference mitigation,''
  {\em IEEE Trans. Wireless Commun.}, vol.~66, pp.~5740--5755, Nov. 2018.

\bibitem{heath2016overview}
{R. W. Heath Jr.}, N.~González-Prelcic, S.~Rangan, W.~Roh, and A.~M. Sayeed,
  ``An overview of signal processing techniques for millimeter wave {MIMO}
  systems,'' {\em IEEE J. Sel. Topics Signal Process.}, vol.~10, pp.~436--453,
  April 2016.

\bibitem{zhang2005variable}
X.~Zhang, A.~F. Molisch, and S.-Y. Kung, ``Variable-phase-shift-based
  {RF}-baseband codesign for {MIMO} antenna selection,'' {\em IEEE Trans.
  Signal Process.}, vol.~53, pp.~4091--4103, Nov. 2005.

\bibitem{alkhateeb2014channel}
A.~Alkhateeb, O.~E. Ayach, G.~Leus, and {R. W. Heath Jr.}, ``Channel estimation
  and hybrid precoding for millimeter wave cellular systems,'' {\em IEEE J.
  Sel. Topics Signal Process.}, vol.~8, pp.~831--846, Oct. 2014.

\bibitem{adhikary2013jointspatial}
A.~Adhikary, J.~Nam, J.~Y. Ahn, and G.~Caire, ``Joint spatial division and
  multiplexing - the large-scale array regime,'' {\em IEEE Trans. Inf. Theory},
  vol.~59, pp.~6441--6463, Oct. 2013.

\bibitem{sohrabi2017hybrid}
F.~Sohrabi and W.~Yu, ``Hybrid analog and digital beamforming for mm{W}ave
  {OFDM} large-scale antenna arrays,'' {\em IEEE J. Sel. Areas Commun.},
  vol.~35, pp.~1432--1443, July 2017.

\bibitem{gao2016energy}
X.~Gao, L.~Dai, S.~Han, C.~L. I, and {R. W. Heath Jr.}, ``Energy-efficient
  hybrid analog and digital precoding for {mmWave} {MIMO} systems with large
  antenna arrays,'' {\em IEEE J. Sel. Areas Commun.}, vol.~34, pp.~998--1009,
  April 2016.

\bibitem{ayach2012thecapacity}
O.~E. Ayach, {R. W. Heath Jr.}, S.~Abu-Surra, S.~Rajagopal, and Z.~Pi, ``The
  capacity optimality of beam steering in large millimeter wave {MIMO}
  systems,'' in {\em Proc. IEEE SPAWC}, pp.~100--104, June 2012.

\bibitem{sadhu2017phasearray}
B.~Sadhu, Y.~Tousi, J.~Hallin, S.~Sahl, S.~Reynolds, Ã.~Renström, K.~Sjögren,
  O.~Haapalahti, N.~Mazor, B.~Bokinge, G.~Weibull, H.~Bengtsson, A.~Carlinger,
  E.~Westesson, J.~E. Thillberg, L.~Rexberg, M.~Yeck, X.~Gu, D.~Friedman, and
  A.~Valdes-Garcia, ``7.2 a 28{GH}z 32-element phased-array transceiver {IC}
  with concurrent dual polarized beams and 1.4 degree beam-steering resolution
  for 5{G} communication,'' in {\em Proc. IEEE ISSCC}, pp.~128--129, Feb. 2017.

\bibitem{lien20175gnr}
S.~Y. Lien, S.~L. Shieh, Y.~Huang, B.~Su, Y.~L. Hsu, and H.~Y. Wei, ``{5G} new
  radio: Waveform, frame structure, multiple access, and initial access,'' {\em
  IEEE Commun. Mag.}, vol.~55, pp.~64--71, June 2017.

\bibitem{niu2015survey}
Y.~Niu, Y.~Li, D.~Jin, L.~Su, and A.~V. Vasilakos, ``A survey of millimeter
  wave communications {(mmWave)} for {5G}: Opportunities and challenges,'' {\em
  Wirel. Netw.}, vol.~21, pp.~2657--2676, Nov. 2015.

\bibitem{maccartney2017rapid}
G.~R. MacCartney, T.~S. Rappaport, and S.~Rangan, ``Rapid fading due to human
  blockage in pedestrian crowds at {5G} millimeter-wave frequencies,'' in {\em
  Proc. IEEE GLOBECOM}, pp.~1--7, Dec. 2017.

\bibitem{va2015basic}
V.~{Va} and R.~W. {Heath}, ``Basic relationship between channel coherence time
  and beamwidth in vehicular channels,'' in {\em Proc. IEEE VTC 2015-Fall},
  pp.~1--5, Sep. 2015.

\bibitem{abari2016millimeter}
O.~Abari, H.~Hassanieh, M.~Rodriguez, and D.~Katabi, ``Millimeter wave
  communications: From point-to-point links to agile network connections,'' in
  {\em Proc. ACM HotNets}, pp.~169--175, ACM, 2016.

\bibitem{son2002highspeed}
B.~S. Son, B.~G. Jo, M.~H. Sunwoo, and Y.~S. Kim, ``A high-speed {FFT}
  processor for {OFDM} systems,'' in {\em Proc. IEEE International Symposium on
  Circuits and Systems. (Cat. No.02CH37353)}, vol.~3, pp.~281--284, May 2002.

\bibitem{3gpp-38-900}
{3rd Generation Partnership Project (3GPP)}, ``{Study on channel model for
  frequency spectrum above 6 {GHz}},'' {\em {TR 38.900 V15.0.0}}, June 2018.

\bibitem{OpenStreetMap}
{OpenStreetMap contributors}, ``{Planet dump retrieved from
  https://planet.osm.org }.'' https://www.openstreetmap.org, 2017.

\bibitem{sun2017nyusim}
S.~Sun, G.~R. MacCartney, and T.~S. Rappaport, ``A novel millimeter-wave
  channel simulator and applications for {5G} wireless communications,'' in
  {\em Proc. IEEE ICC}, pp.~1--7, May 2017.

\end{thebibliography}

\end{document}